\documentclass[11pt]{llncs}
\usepackage{amsmath}
\usepackage{amssymb}
\usepackage{cite}
\usepackage[latin1]{inputenc}   
\usepackage[T1]{fontenc}  
\usepackage{comment} 
\begin{document}
\frontmatter          
\pagestyle{headings}  
\addtocmark{Hamiltonian Mechanics} 

\mainmatter              
\title{On Polynomial Kernelization of {\normalfont\scshape $\mathcal{H}$\textbf{-free Edge Deletion}}}
%
%
\author{N. R. Aravind\hspace{0.3in} R. B. Sandeep\thanks{supported by TCS Research Scholarship}\hspace{0.3in} Naveen Sivadasan}
%
%

%
\institute{Department of Computer Science \& Engineering\\
Indian Institute of Technology Hyderabad, India\\
\email{\{aravind,cs12p0001,nsivadasan\}\makeatletter@\makeatother iith.ac.in}}

\maketitle              

\begin{abstract}
For a set of graphs $\mathcal{H}$, the \textsc{$\mathcal{H}$-free Edge Deletion} problem asks to find whether there exist at most $k$ edges in the input graph whose deletion results in a graph without any induced copy of $H\in\mathcal{H}$. 
In \cite{cai1996fixed}, it is shown that the problem is fixed-parameter tractable if $\mathcal{H}$ is of finite cardinality. However, it is proved in \cite{cai2013incompressibility} that if $\mathcal{H}$ is a singleton set containing $H$, 
  for a large class of $H$, there exists no polynomial kernel unless $coNP\subseteq NP/poly$. 
In this paper, we present a polynomial kernel for this problem for any fixed finite set $\mathcal{H}$ of connected graphs and when the input graphs are of bounded degree. 
We note that there are \textsc{$\mathcal{H}$-free Edge Deletion} problems which remain NP-complete even for the bounded degree input graphs, for example \textsc{Triangle-free Edge Deletion}\cite{brugmann2009generating} and \textsc{Custer Edge Deletion($P_3$-free Edge Deletion)}\cite{komusiewicz2011alternative}.
When $\mathcal{H}$ contains $K_{1,s}$, we obtain a stronger result - a polynomial kernel for $K_t$-free input graphs (for any fixed $t> 2$).
We note that for $s>9$, there is an incompressibility result for \textsc{$K_{1,s}$-free Edge Deletion}  for general graphs \cite{cai2012polynomial}.
Our result provides first polynomial kernels for \textsc{Claw-free Edge Deletion} and \textsc{Line Edge Deletion} for $K_t$-free input graphs which are NP-complete even for $K_4$-free graphs\cite{yannakakis1981edge} and were raised as open problems in \cite{cai2013incompressibility,open2013worker}.
\end{abstract}
\section{Introduction}
For a graph property $\Pi$, the \textsc{$\Pi$ Edge Deletion} problem asks whether there exist at most $k$ edges such that deleting them from the input graph
 results in a graph with property $\Pi$. Numerous studies have been done on edge deletion problems from 1970s onwards dealing with various aspects 
such as hardness
\cite{yannakakis1981edge,garey1974some,el1988complexity,natanzon2001complexity,alon2009additive,goldberg1995four,margot1994some,karp72,brugmann2009generating,shamir2004cluster},
polynomial-time algorithms\cite{natanzon2001complexity,hadlock1975finding,shamir2004cluster},
approximability \cite{natanzon2001complexity,alon2009additive,shamir2004cluster},
fixed-parameter tractability\cite{cai1996fixed,gramm2003graph},
polynomial problem kernels\cite{guo2007problem,brugmann2009generating,paul2012non,gramm2003graph} and 
incompressibility\cite{cai2013incompressibility,kratsch2009two,cai2012polynomial}.

There are not many generalized results on the NP-completeness of edge deletion problems. This is in contrast with the classical result by 
Lewis and Yannakakis\cite{lewis1980node} on the vertex counterparts which says that \textsc{$\Pi$ Vertex Deletion} problems are NP-complete if $\Pi$ is non-trivial
and hereditary on induced subgraphs. By a result of Cai\cite{cai1996fixed}, the \textsc{$\Pi$ Edge Deletion} problem is
fixed-parameter tractable for any hereditary property $\Pi$ that is characterized by a finite set of forbidden induced subgraphs. We observe that polynomial problem kernels have been found only for a few parameterized \textsc{$\Pi$ Edge Deletion} problems. 

In this paper, we study a subset of \textsc{$\Pi$ Edge Deletion} problems known as \textsc{$\mathcal{H}$-free Edge Deletion} problems where 
$\mathcal{H}$ is a set of graphs. The objective is 
to find whether there exist at most $k$ edges in the input graph such that deleting them results in a graph with no induced copy of $H\in\mathcal{H}$. In the natural
parameterization of this problem, the parameter is $k$. 
In this paper, we give a polynomial problem kernel for parameterized version of \textsc{$\mathcal{H}$-free Edge Deletion} where 
$\mathcal{H}$ is any fixed finite set of 
connected graphs and when the input graphs are of bounded degree. In this context, we note that \textsc{Triangle-free Edge Deletion}\cite{brugmann2009generating} and \textsc{Custer Edge Deletion($P_3$-free Edge Deletion)}\cite{komusiewicz2011alternative} are NP-complete even for bounded degree 
input graphs. We also note that, under the complexity theoretic assumption $coNP\not\subseteq NP/poly$, there exist no polynomial problem kernels for the \textsc{$H$-free Edge Deletion} problems when $H$ is 3-connected but not complete, or when $H$ is a path or cycle of at least 4 edges\cite{cai2013incompressibility}. When the input graph has maximum degree at most $\Delta$ and if the maximum diameter of graphs in $\mathcal{H}$ is $D$, then the number of vertices in the kernel we obtain is at most $2\Delta^{2D+1}\cdot k^{pD+1}$ where $p=\log_{\frac{2\Delta}{2\Delta-1}}\Delta$. Our kernelization consists of a single rule which removes vertices of the input graph that are `far enough' from all induced $H\in \mathcal{H}$ in $G$. 

When $\mathcal{H}$ contains $K_{1,s}$, we obtain a stronger result - a polynomial kernel for $K_t$-free input graphs (for any fixed $t> 2$).
Let $s>1$ be the least integer such that $K_{1,s}\in \mathcal{H}$. Then the number of 
vertices in the kernel we obtain is at most $8d^{3D+1}\cdot k^{pD+1}$ where $d=R(s,t-1)-1$, $R(s,t-1)$ is the Ramsey number and $p=\log_{\frac{2d}{2d-1}}d$.
We note that \textsc{Claw-free Edge Deletion} and \textsc{Line Edge Deletion} are NP-complete even for $K_4$-free input graphs\cite{yannakakis1981edge}. 
As a corollary of our result, we obtain the first polynomial kernels for these problems when the input graphs are $K_t$-free for any fixed 
$t>2$. The existence of a polynomial kernel for \textsc{Claw-free Edge Deletion} and \textsc{Line Edge Deletion} were raised as open problems in \cite{cai2013incompressibility,open2013worker}.
We note that for $s>9$, there is an incompressibility result for \textsc{$K_{1,s}$-free Edge Deletion}  for general graphs \cite{cai2012polynomial}. 
\subsection{Related Work}
Here, we give an overview of various results on edge deletion problems.
\paragraph{NP-completeness:} It has been proved that \textsc{$\Pi$ Edge Deletion} problems
are NP-complete if $\Pi$ is one of the following properties: without cycle of any fixed length $l\geq 3$, without any cycle of length at most $l$ for any fixed $l\geq 4$, connected with maximum degree $r$ for every fixed
$r\geq 2$, outerplanar, line graph, bipartite, comparability\cite{yannakakis1981edge}, claw-free (implicit in the proof of NP-completeness of the \textsc{Line Edge Deletion} problem in \cite{yannakakis1981edge}), $P_l$-free for any fixed $l\geq 3$\cite{el1988complexity}, circular-arc, chordal, chain, perfect, split, AT-free\cite{natanzon2001complexity}, interval\cite{goldberg1995four},
threshold\cite{margot1994some} and 
complete\cite{karp72}.

\paragraph{Fixed-parameter Tractability and Kernelization:} Cai proved in \cite{cai1996fixed} that parameterized \textsc{$\Pi$ Edge Deletion} problem is fixed-parameter tractable if $\Pi$ is a hereditary property characterized by a finite set of forbidden induced subgraphs. Hence 
\textsc{$\mathcal{H}$-free Edge Deletion} is fixed-parameter tractable for any finite set of graphs $\mathcal{H}$.
Polynomial problem
kernels are known for chain, split, threshold\cite{guo2007problem}, triangle-free\cite{brugmann2009generating}, cograph\cite{paul2012non} and cluster\cite{gramm2003graph} edge deletions. It is proved in \cite{cai2013incompressibility} that for 3-connected $H$, $\textsc{H-free Edge Deletion}$ admits no polynomial kernel if and only if $H$ is not a complete graph, under the assumption $coNP\not\subseteq NP/poly$. Under the same assumption, it is proved in \cite{cai2013incompressibility} that for $H$ being a path or cycle, $\textsc{H-free Edge Deletion}$ admits no polynomial kernel if and only if $H$ has at least 4 edges.
Unless $NP\subseteq coNP/poly$, $\textsc{H-free Edge Deletion}$ admits no polynomial kernel if $H$ is $K_1 \times\ (2K_1\cup 2K_2)$ \cite{kratsch2009two}.

\section{Preliminaries and Basic Results}
We consider only simple graphs. For a set of graphs $\mathcal{H}$, a graph $G$ is $\mathcal{H}$-free if there is no induced copy of $H\in\mathcal{H}$ in $G$. For $V'\subseteq V(G)$,
$G\setminus V'$ denotes the graph $(V(G)\setminus V', E(G)\setminus E')$ where $E'\subseteq E(G)$ is the set of edges incident to vertices in $V'$. Similarly, for $E'\subseteq E(G)$,
$G\setminus E'$ denotes the graph $(V(G), E(G)\setminus E')$. For any edge set $E'\subseteq E(G)$, $V_{E'}$ denotes the set of vertices 
incident to the edges in $E'$. For any $V'\subseteq V(G)$, the closed neighbourhood of $V'$, $N_G[V']=\{v: v\in V'\ \text{or}\ (u,v)\in E(G)\ \text{for some}\ u\in V'\}$. In a graph $G$, distance from a vertex $v$ to a set of vertices $V'$ is the shortest among the distances from $v$ to the vertices in $V'$. 

A parameterized problem is \textit{fixed-parameter tractable}(FPT) if there exists an algorithm to solve it which runs in time $O(f(k)n^c)$ where $f$ is a computable function, $n$ is the input size, $c$ is a constant and $k$ is the parameter. The idea is to solve the problem efficiently for small parameter values. 
A related notion is \textit{polynomial kernelization} where the parameterized problem instance is reduced in polynomial (in $n+k$) time to a polynomial (in $k$) sized 
instance of the same problem called \textit{problem kernel} such that the 
original instance is a yes-instance if and only if the problem kernel is a yes-instance. We refer to \cite{downey2013fundamentals} for an exhaustive treatment on these topics. 
A kernelization rule is \textit{safe} if the answer to the problem instance does not change after the application of the rule. 

In this paper, we consider \textsc{$\mathcal{H}$-free Edge Deletion}\footnote{we leave the prefix `parameterized' henceforth as it is evident from the context} which is defined as given below. 

\vspace{0.1cm}

\fbox{\parbox{11cm}{
\textsc{$\mathcal{H}$-free Edge Deletion}\\
\textbf{Instance}: A graph $G$ and a positive integer $k$.\\
\textbf{Problem}: Does there exist $E'\subseteq E(G)$ with $|E'|\leq k$ such that $G\setminus E'$ does not contain $H\in\mathcal{H}$ as an induced subgraph.\\
\textbf{Parameter}: $k$
}}

\vspace{0.1cm}

We define an \textit{$\mathcal{H}$ deletion set (HDS)} of a graph $G$ as a set $M\subseteq E(G)$ such that $G\setminus M$ is $\mathcal{H}$-free. The \textit{minimum $\mathcal{H}$ deletion set (MHDS)} is an HDS with smallest cardinality. 
We define a partition of an MHDS $M$ of $G$ as follows.

$M_1=\{e: e\in M\ \text{and}\ e\ \text{is part of an induced}\ H\in \mathcal{H}\ \text{in}\ G\}.$

$M_j= \{e: e\in M\setminus \bigcup_{i=1}^{i=j-1}M_i\ \text{and}\ e\ \text{is part of an induced}\ H\in \mathcal{H}\ \text{in}\ G\setminus \bigcup_{i=1}^{i=j-1}M_i\}, \text{for}\ j>1$.

We define the \textit{depth} of an MHDS $M$ of $G$, denoted by $l_M$, as the least integer such that $|M_i|>0$ for all $1\leq i\leq l_M$   and $|M_i|=0$ for all $i>l_M$. Proposition~\ref{prop:depth} shows that this notion is well defined.

\begin{proposition}\label{prop:depth}
\begin{enumerate}
\item $\{M_j\}$ forms a partition of $M$.
\item There exists $l_M\geq 0$ such that $|M_i|>0$ for $1\leq i\leq l_M$ and $|M_i|=0$ for $i > l_M$.
\end{enumerate}
\end{proposition}
\begin{proof}
If $i\not = j$ and $M_i$ and $M_j$ are nonempty, then $M_i\cap M_j=\emptyset$. For $i\geq 1$, $M_i\subseteq M$. Assume there is an edge $e\in M$ and $e\notin \bigcup M_j$. Delete all edges in $\bigcup M_j$ from $G$. What remains is an $\mathcal{H}$-free graph. As $M$ is an MHDS, there can not exist such an edge $e$. Now let $j$ be the smallest integer such that $M_j$ is empty. Then from definition, for all $i>j$, $|M_i|=0$. Therefore 
$l_M=j-1$.
\end{proof}\qed

We observe that for an $\mathcal{H}$-free graph, the only MHDS $M$ is $\emptyset$ and hence $l_M=0$. For an MHDS $M$ of $G$ with a depth $l_M$, 
we define the following terms.

$S_j=\bigcup_{i=j}^{i=l_M}M_j$ for $1\leq j\leq l_M+1$.   

$T_j= M\setminus S_{j+1}$ for $0\leq j\leq l_M$.

$V_{\mathcal{H}}(G)$ is the set of all vertices part of some induced $H\in \mathcal{H}$ in $G$.

We observe that $S_1=T_{l_M}=M$, $S_{l_M} = M_{l_M}$, $T_1= M_1$ and $S_{l_M+1}=T_0=\emptyset$.
\begin{proposition}\label{prop:stilleasy}
For a graph $G$, let $E'\subseteq E(G)$ such that at least one edge in every induced $H\in \mathcal{H}$ in $G$ is in $E'$. Then, at least one vertex in every induced $H\in \mathcal{H}$ in 
$G\setminus E'$ is in $V_{E'}$.
\end{proposition}
\begin{proof}
Assume that there exists an induced $H\in \mathcal{H}$ in $G\setminus E'$ with the vertex set $V'$. For a contradiction, assume that $|V'\cap V_{E'}|=0$. Then, $V'$ induces a copy of $H$ in $G$. Hence, $E'$ must contain some of its edges. 
\end{proof}\qed

\begin{lemma}\label{lem:dist}
Let $G$ be the input graph of an \textsc{$\mathcal{H}$-free Edge Deletion} problem instance where $\mathcal{H}$ is a set of connected graphs with diameter at most $D$. Let $M$ be an MHDS of $G$. Then, every vertex in $V_M$ is at a distance at most $(l_M-1)D$ from $V_{\mathcal{H}}(G)$ in $G$ .
\end{lemma}
\begin{proof}
For $2\leq j\leq l_M$, from definition, at least one edge in every induced $H\in\mathcal{H}$ in $G\setminus T_{j-2}$ is in $M_{j-1}$. Hence by
Proposition~\ref{prop:stilleasy}, at least one vertex in every induced $H\in \mathcal{H}$ in $G\setminus T_{j-1}$ is in $V_{M_{j-1}}$. By definition, every vertex in $V_{M_j}$ is part of some induced $H\in\mathcal{H}$ in $G\setminus T_{j-1}$. This implies every vertex in 
$V_{M_j}$ is at a distance at most $D$ from $V_{M_{j-1}}$. Hence every vertex in $V_{M_{l_M}}$ is at a distance at most $(l_M-1)D$ from $V_{M_1}$. 
By definition, $V_{M_1}\subseteq V_{\mathcal{H}}(G)$. Hence the proof.
\end{proof}\qed

\begin{lemma}\label{lem:step}
Let $G$ be a graph with maximum degree at most $\Delta$ and $M$ be an MHDS of $G$.  Then, for $1\leq j\leq l_M$, $(2\Delta-1)\cdot |M_j|\geq |S_{j+1}|$.
\end{lemma}
\begin{proof}
For $1\leq j\leq l_M$, from definition, $M_j$ has at least one edge from every induced $H\in \mathcal{H}$ in $G\setminus T_{j-1}$. Let 
$M'_j$ be the set of edges incident to vertices in $V_{M_j}$ in $G\setminus T_{j-1}$. We observe that $(G\setminus T_{j-1})\setminus M'_j$ is $\mathcal{H}$-free and hence $|T_{j-1}\cup M'_j|$ is an HDS of $G$. Clearly, $|M'_j|\leq \Delta |V_{M_j}|\leq 2\Delta |M_j|$. Since $M$ is an MHDS, 
$|T_{j-1}\cup M'_j| = |T_{j-1}|+|M'_j|\geq |M|=|T_{j-1}|+|S_j|$. Therefore
$|M'_j|\geq|S_j|$. Hence, $2\Delta |M_j|\geq|S_j|=|M_j|+|S_{j+1}|$.
\end{proof}\qed

Now we give an upper bound for the depth of an MHDS in terms of its size and maximum degree of the graph.
\begin{lemma}\label{lem:lower}
Let $M$ be an MHDS of $G$. If the maximum degree of $G$ is at most $\Delta>0$, 
then $l_M\leq 1+\log_{\frac{2\Delta}{2\Delta-1}}|M|$.
\end{lemma}

\begin{proof}
The statement is clearly true when $l_M\leq 1$. Hence assume that $l_M\geq 2$. The result follows from repeated application of Lemma~\ref{lem:step}.
\begin{align*}
|M| &= |S_1|=|M_1|+|S_2|\geq \frac{|S_2|}{2\Delta-1}+|S_2|\\
  &\geq |S_{l_M}|\left(\dfrac{2\Delta}{2\Delta-1}\right)^{l_M-1}\\
  &\geq \left(\dfrac{2\Delta}{2\Delta-1}\right)^{l_M-1} \hspace{0.5cm} \text{[$\because |S_{l_M}|\geq 1$]}
\end{align*}
\end{proof}\qed

\begin{corollary}\label{cor:l}
Let $(G,k)$ be a yes-instance of \textsc{$\mathcal{H}$-free Edge Deletion} where $G$ has maximum degree at most $\Delta>0$. For any MHDS $M$
of $G$, 
$l_M\leq 1 + \log_{\frac{2\Delta}{2\Delta-1}}k$. 
\end{corollary}\qed

\begin{lemma}\label{lem:gensafe}
Let $\mathcal{H}$ be a set of connected graphs with diameter at most $D$. Let $V'\supseteq V_{\mathcal{H}}(G)$ and let $c\geq 0$.
Let $G'$ be obtained by removing vertices of $G$ at a distance more than $c+D$ from $V'$. 
Furthermore, assume that if $G'$ is a yes-instance then there exists an MHDS $M'$ of $G'$ such that every vertex in $V_{M'}$ is at a distance at most $c$ from  $V'$ in $G'$. Then $(G,k)$ is a 
yes-instance if and only if $(G',k)$ is a yes-instance of \textsc{$\mathcal{H}$-free Edge Deletion}.
\end{lemma}
\begin{proof}
Let $G$ be a yes-instance with an MHDS $M$. Then $M'=M\cap E(G')$ is an HDS of $G'$ such that $|M'|\leq k$. Conversely, let $G'$ be a yes-instance. By the assumption, there exists an
MHDS $M'$ of $G'$ such that every vertex in $V_{M'}$ is at a distance at most $c$ from $V'$ in $G'$. 
We claim that $M'$ is an MHDS of $G$. For contradiction, assume $G\setminus M'$ has an induced $H\in \mathcal{H}$ with a vertex set $V''$. As $G$ and $G'$ has same set of 
induced copies of graphs in $\mathcal{H}$, at least one edge 
in every induced copy of graphs in $\mathcal{H}$ in $G$ is in $M'$. Then, by Proposition~\ref{prop:stilleasy}, at least one vertex in $V''$ is in $V_{M'}$. 
We observe that for every vertex in $G'$ the distance from $V'$ is same in $G$ and $G'$. Hence every vertex in $V''$ is at a distance at most $c+D$ from $V'$ in $G$. Then, $V''$ induces a copy of $H$ in $G'\setminus M'$ which is a 
contradiction. 
\end{proof}\qed

\begin{lemma}\label{lem:genbound}
Let $G$ be a graph and let $d>1$ be a constant. Let $V'\subseteq V(G)$ such that all vertices in $G$ with degree more than $d$ is in $V'$. Partition $V'$ into $V_1$ and $V_2$ such that $V_1$ contains all the vertices in $V'$ with degree at most $d$ and $V_2$ contains all the vertices with degree more than $d$. If every vertex in $G$ is at a distance at most $c>0$ from  $V'$, then $|V(G)|\leq |V_1|\cdot d^{c+1} + |N_G(V_2)|\cdot d^{c}$.
\end{lemma}
\begin{proof}
To enumerate the number of vertices in $G$, consider the $d$-ary breadth first trees rooted at vertices in $V_1$ and in $N_G[V_2]$. 
\begin{align*}
|V(G')| &\leq |V_1|\left(\dfrac{d^{c+1}-1}{d-1}\right) + |N_G[V_2]|\left(\dfrac{d^{c}-1}{d-1}\right)\\
  &\leq |V_1|d^{c+1} + |N_G[V_2]|d^{c}
\end{align*}
\end{proof}\qed

\section{Polynomial Kernels}
In this section, we assume that $\mathcal{H}$ is a fixed finite set of connected graphs with diameter at most $D$. First we devise an algorithm to obtain polynomial kernel for \textsc{$\mathcal{H}$-free Edge Deletion} for bounded degree input graphs. Then we prove a stronger result - a polynomial kernel for $K_t$-free input graphs (for some fixed $t>2$) when $\mathcal{H}$ contains $K_{1,s}$ for some $s>1$. 

We assume that the input graph $G$ has maximum degree at most $\Delta>1$ and $G$ has at least
one induced copy of $H$. We observe that if these conditions are not met, obtaining polynomial kernel is trivial.

Now we state the kernelization rule which is the single rule in the kernelization. 
\begin{description}
\item[Rule 0:] Delete all vertices in $G$ at a distance more than $(1+\log_{\frac{2\Delta}{2\Delta-1}}k)D$ from $V_{\mathcal{H}}(G)$.
\end{description}

We note that the rule can be applied efficiently with the help of breadth first search from vertices in $V_{\mathcal{H}}(G)$. Now we prove the 
safety of the rule.

\begin{lemma}\label{lem:safe}
Rule $0$ is safe.
\end{lemma}
\begin{proof}
Let $G'$ be obtained from $G$ by applying Rule $0$. Let $M'$ be an MHDS of $G'$. If $G'$ is a yes-instance, then by Lemma~\ref{lem:dist} and Corollary~\ref{cor:l},
every vertex in $V_{M'}$ is at a distance at most $D\log_{\frac{2\Delta}{2\Delta-1}}k$ from $V_{\mathcal{H}}(G')$. 
Hence, we can apply Lemma~\ref{lem:gensafe} with $V'=V_{\mathcal{H}}(G)$ and $c=D\log_{\frac{2\Delta}{2\Delta-1}}k$. 
\end{proof}\qed

\begin{lemma}\label{lem:bound}
Let $(G,k)$ be a yes-instance of \textsc{$\mathcal{H}$-free Edge Deletion}. Let $G'$ be obtained by one application of Rule 0 on $G$. Then, $|V(G')|\leq (2\Delta^{2D+1}\cdot k^{pD+1})$ where $p=\log_{\frac{2\Delta}{2\Delta-1}}\Delta$.
\end{lemma}
\begin{proof}
Let $M$ be an MHDS of $G$ such that $|M|\leq k$. We observe that every vertex in 
$V_{\mathcal{H}}(G)$ is at a distance at most $D$ from $V_{M_1}$ in $G$. Hence, by construction, every vertex in $G'$ is at a distance at most 
$(2+\log_{\frac{2\Delta}{2\Delta-1}}k)D$ from $V_{M_1}$ in $G$ and in $G'$. We note that $|V_{M_1}|\leq 2k$. To enumerate the number of vertices in $G'$, we apply Lemma~\ref{lem:genbound} with
$V'=V_{M_1}$, $c=(2+\log_{\frac{2\Delta}{2\Delta-1}}k)D$ and $d=\Delta$.
\begin{align*}
|V(G')| &\leq 2k\Delta^{(2+\log_{\frac{2\Delta}{2\Delta-1}}k)D+1}\\
  &\leq 2\Delta^{2D+1}\cdot k^{pD+1}
\end{align*}
\end{proof}\qed

Now we present the algorithm to obtain a polynomial kernel. The algorithm applies Rule 0 on the input graph and according to the number of vertices in the resultant graph it returns the resultant graph or a trivial no-instance.

\vspace{0.1cm}

\fbox{\parbox{11cm}{
Kernelization for \textsc{$\mathcal{H}$-free Edge Deletion}\\
($\mathcal{H}$ is a finite set of connected graphs with maximum diameter $D$)\\
Input:$(G,k)$ where $G$ has maximum degree at most $\Delta$.
\begin{description}
\item[1.] Apply Rule 0 on $G$ to obtain $G'$.
\item[2.] If the number of vertices in $G'$ is more than $2\Delta^{2D+1}\cdot k^{pD+1}$ where $p=\log_{\frac{2\Delta}{2\Delta-1}}\Delta$, then return
a trivial no-instance $(H,0)$ where $H$ is the graph with minimum number of vertices in $\mathcal{H}$. Else return $(G',k)$.
\end{description}
}}

\begin{theorem}\label{thm:correct}
The kernelization for \textsc{$\mathcal{H}$-free Edge Deletion} returns a kernel with the number of vertices at most $2\Delta^{2D+1}\cdot k^{pD+1}$ where $p=\log_{\frac{2\Delta}{2\Delta-1}}\Delta$.
\end{theorem}
\begin{proof}
Follows from Lemma~\ref{lem:safe} and Lemma~\ref{lem:bound} and the observation that the number of vertices in the trivial no-instance is at most $2\Delta^{2D+1}\cdot k^{pD+1}$.
\end{proof}\qed

\subsection{A stronger result for a restricted case}
Here we give a polynomial kernel for \textsc{$\mathcal{H}$-free Edge Deletion} when $\mathcal{H}$ is a fixed finite set of connected graphs and contains a $K_{1,s}$ for some $s> 1$ and when the input graphs are $K_t$-free, for any fixed 
$t> 2$. 

It is proved in \cite{le2008stable} that the maximum degree of a $\{\text{claw},K_4\}$-free graph is at most 5. We give a straight forward
 generalization of this result for $\{K_{1,s},K_t\}$-free graphs. Let $R(s,t)$ denote the Ramsey number. 
Remember that the Ramsey number $R(s,t)$ is the least integer such that every graph on $R(s,t)$ vertices has either an independent set of order
$s$ or a complete subgraph of order $t$.

\begin{lemma}\label{lem:ramsey}
For integers $s> 1,t> 1$, any $\{K_{1,s}, K_t\}$-free graph has maximum degree at most $R(s,t-1)-1$.
\end{lemma}
\begin{proof}
Assume $G$ is $\{K_{1,s}, K_t\}$-free. For contradiction, assume $G$ has a vertex $v$ of degree at least $R(s,t-1)$. By the definition of the Ramsey number there exist at least $s$ mutually non-adjacent vertices or $t-1$ mutually adjacent vertices in the neighborhood of $v$. Hence
there exist either an induced $K_{1,s}$ or an induced $K_t$ in $G$.
\end{proof}\qed

We modify the proof technique used for devising polynomial kernelization for \textsc{$\mathcal{H}$-free Edge Deletion} for bounded degree graphs to obtain polynomial
kernelization for $K_t$-free input graphs for the case when $\mathcal{H}$ contains $K_{1,s}$ for some $s>1$. 

Let $s> 1$ be the least integer such that $\mathcal{H}$ contains $K_{1,s}$. Let $t> 2$, $G$ be $K_{t}$-free and $M$ be an MHDS of $G$. Let $d=R(s,t-1)-1$. Let $D$ be the maximum diameter of graphs in $\mathcal{H}$. We define the following.

\begin{description}
\item[$M_0=$]$\{e: e\in M\ \text{and}\ e\ \text{is incident to a vertex with degree at least}\ d+1\}$.
\item[$V_R(G)=$]$\{v: v\in V(G)\ \text{and}\ v\ \text{has degree at least $d+1$ in $G$}\}$.
\end{description}

\begin{lemma}\label{lem:max}
$G\setminus M_0$ has degree at most $d$ and every vertex in $G$ with degree at least $d+1$ is incident to at least one edge in $M_0$.
\end{lemma}
\begin{proof}
As $G\setminus M$ is $\{K_{1,s}, K_t\}$-free and every edge in $M$ which is incident to at least one vertex of degree at least $d+1$ is in $M_0$, the result follows from Lemma~\ref{lem:ramsey}.
\end{proof}\qed

\begin{lemma}\label{lem:stdist}
Let $M$ be an MHDS of $G$. Let $M'=M\setminus M_0$ and $G'=G\setminus M_0$. Then, $M'$ is an MHDS of $G'$ and every vertex in $V_M$ is at a distance at most $Dl_{M'}$ from $V_{\mathcal{H}}(G)\cup V_R(G)$ in $G$ .
\end{lemma}
\begin{proof}
It is straight forward to verify that $M'$ is an MHDS of $G'$. 
By Lemma~\ref{lem:dist}, every vertex in $V_{M'}$ is at a distance at most $(l_{M'}-1)D$ from $V_{\mathcal{H}}(G')$ in $G'$. Every induced $H\in\mathcal{H}$ in $G'$ is either an induced $H$ in $G$ or formed by deleting $M_0$ from $G$. Therefore, every vertex in
$V_{\mathcal{H}}(G')$ is at a distance at most $D$ from $V_{\mathcal{H}}(G)\cup V_R(G)$ in $G'$. Hence, every vertex in $V_{M'}$ is at a distance at most $Dl_{M'}$ from  $V_{\mathcal{H}}(G)\cup V_R(G)$ in $G'$. The result follows from the fact $M=M'\cup M_0$.
\end{proof}\qed

The single rule in the kernelization is:

\begin{description}
\item[Rule 1:] Delete all vertices in $G$ at a distance more than $(2+\log_{\frac{2d}{2d-1}}k)D$ from $V_{\mathcal{H}}(G)\cup V_R(G)$ where 
$d=R(s,t-1)-1$.
\end{description}

\begin{lemma}\label{lem:stsafe}
Rule $1$ is safe.
\end{lemma}
\begin{proof}
Let $G'$ be obtained from $G$ by applying Rule $1$. Let $M'$ be an MHDS of $G'$. If $G'$ is a yes-instance, then by Lemma~\ref{lem:stdist} and Corollary~\ref{cor:l}, every vertex in $V_{M'}$ is at a distance at most 
$D(1+\log_{\frac{2d}{2d-1}}k)$ from $V_{\mathcal{H}}(G')\cup V_R(G')$ in $G'$. We note that $V_{\mathcal{H}}(G)=V_{\mathcal{H}}(G')$ and $V_R(G)=V_R(G')$. Hence, we can apply Lemma~\ref{lem:gensafe} with $V'=V_{\mathcal{H}}(G)\cup V_R(G)$, $c=D(1+\log_{\frac{2d}{2d-1}}k)$ and $d=R(s,t-1)-1$. 
\end{proof}\qed

\begin{lemma}\label{lem:stbound}
Let $(G,k)$ be a yes-instance of \textsc{$\mathcal{H}$-free Edge Deletion} where $G$ is $K_t$-free. Let $G'$ be obtained by one application of Rule 1 on $G$. Then, $|V(G')|\leq 8d^{3D+1}\cdot k^{pD+1}$ where $p=\log_{\frac{2d}{2d-1}}d$.
\end{lemma}
\begin{proof}
Let $M$ be an MHDS of $G$ such that $|M|\leq k$. We observe that every vertex in $V_{\mathcal{H}}(G)$ is at a distance at most $D$ from $V_{M_1}$ in $G$. Hence, by construction, every vertex in $G'$ is at a distance at most
$D(3+\log_{\frac{2d}{2d-1}}k)$ from  $V_{M_1}\cup V_R(G)$. Clearly $|V_{M_1}|\leq 2k$. Using Lemma~\ref{lem:max} we obtain $|N[V_R(G)]|\leq 2k(d+2)$. To enumerate the number of vertices in $G'$, we apply Lemma~\ref{lem:genbound} with
$V'=V_{M_1}\cup V_R(G)$, $c=D(3+\log_{\frac{2d}{2d-1}}k)$ and $d=R(s,t-1)-1$.
\begin{align*}
|V(G')| &\leq 2kd^{D(3+\log_{\frac{2d}{2d-1}}k)+1}+2k(d+2)d^{D(3+\log_{\frac{2d}{2d-1}}k)}\\
  &\leq 8d^{3D+1}\cdot k^{pD+1}
\end{align*}
\end{proof}\qed

Now we present the algorithm.
\vspace{0.1cm}

\fbox{\parbox{11cm}{
Kernelization for \textsc{$\mathcal{H}$-free Edge Deletion}\\
($\mathcal{H}$ contains $K_{1,s}$ for some $s>1$)\\
Input:$(G,k)$ where $G$ is $K_{t}$-free for some fixed $t>2$.\\
Let $s>1$ be the least integer such that $\mathcal{H}$ contains $K_{1,s}$.
\begin{description}
\item[1.] Apply Rule 1 on $G$ to obtain $G'$.
\item[2.] If the number of vertices in $G'$ is more than $8d^{3D+1}\cdot k^{pD+1}$ where $d=R(s,t-1)-1$ and $p=\log_{\frac{2d}{2d-1}}d$, then return
a trivial no-instance $(K_{1,s},0)$. Else return $(G',k)$.
\end{description}
}}
\vspace{0.1cm}

For practical implementation, we can use any specific known upper bound for $R(s,t-1)$ or the general upper bound 
$\binom{s+t-3}{s-1}$.
\begin{theorem}\label{thm:stcorrect}
The kernelization for \textsc{$\mathcal{H}$-free Edge Deletion} when $K_{1,s}\in \mathcal{H}$ and the input graph is $K_t$-free returns a kernel with the number of vertices at most $8d^{1+3D}\cdot k^{1+pD}$ where $d=R(s,t-1)-1$ and $p=\log_{\frac{2d}{2d-1}}d$.
\end{theorem}
\begin{proof}
Follows from Lemma~\ref{lem:stsafe} and Lemma~\ref{lem:stbound}.
\end{proof}\qed

It is known that line graphs are characterized by a finite set of connected forbidden induced subgraphs including a claw ($K_{1,3}$). Both
\textsc{Claw-free Edge Deletion} and \textsc{Line Edge Deletion} are NP-complete even for $K_4$-free graphs\cite{yannakakis1981edge}.
\begin{corollary}
\textsc{Claw-free Edge Deletion} and \textsc{Line Edge Deletion} admit polynomial kernels for $K_t$-free input graphs for any fixed 
$t>3$.
\end{corollary}\qed

We observe that the kernelization for \textsc{$\mathcal{H}$-free Edge Deletion} when $K_{1,s}\in\mathcal{H}$ and the input graph is $K_t$-free works
for the case when $K_t\in \mathcal{H}$ and the input graph is $K_{1,s}$-free.

\begin{theorem}
\textsc{$\mathcal{H}$-free Edge Deletion} admits polynomial kernelization when $\mathcal{H}$ is a finite set of connected graphs, $K_t\in\mathcal{H}$ for some $t>2$ and the input graph is $K_{1,s}$-free for some fixed $s>1$.
\end{theorem}

\section{Concluding Remarks} Our results may give some insight towards a dichotomy theorem on incompressibility of \textsc{$\mathcal{H}$-free Edge Deletion} raised as an open problem in \cite{cai2013incompressibility}.
We conclude with an open problem:
does \textsc{$\mathcal{H}$-free Edge Deletion} admit polynomial kernel for planar input graphs?
\bibliographystyle{plain}
\bibliography{ed}
%
\end{document}